\documentclass[11pt,a4paper]{article}
\usepackage[utf8]{inputenc}
\usepackage[T1]{fontenc}
\usepackage[margin=25mm]{geometry}
\usepackage{amsmath,amssymb,bm}
\usepackage{graphicx}
\usepackage{authblk}
\usepackage[numbers,sort&compress]{natbib}
\newcommand{\onlinecite}[1]{\citenum{#1}}
\usepackage[hidelinks]{hyperref}
\newtheorem{proposition}{Proposition}
\newtheorem{corollary}{Corollary}
\newcommand{\qedmark}{\unskip\nobreak\hfill\ensuremath{\blacksquare}}
\newcommand{\R}{\mathbb{R}}
\newcommand{\norm}[1]{\left\lVert#1\right\rVert}
\hypersetup{pdftitle={Extrapolating Learned Dynamics across Bifurcations Based on Generalized Synchronization},pdfauthor={Lingfeng Chen and Masanobu Inubushi}}

\title{Extrapolating learned dynamics across bifurcations\texorpdfstring{\\}{ }based on generalized synchronization}
\author[1]{Lingfeng Chen}
\author[1,2]{Masanobu Inubushi\thanks{Corresponding author: \href{mailto:inubushi@rs.tus.ac.jp}{inubushi@rs.tus.ac.jp}}}
\affil[1]{Department of Applied Mathematics, Tokyo University of Science, 1-3, Kagurazaka, Shinjuku, Tokyo 162-8601, Japan}
\affil[2]{Graduate School of Engineering Science, The University of Osaka, 1-3 Machikaneyama, Toyonaka, Osaka 560-8531, Japan}
\date{}
\begin{document}
\maketitle
\begin{abstract}
Predicting dynamics across bifurcations by extrapolating trained models is crucial in diverse fields of physics. Generalized synchronization provides a framework for understanding dynamical reconstruction by neural networks, especially in the context of reservoir computing. Building on this framework, we derive a formula combining pretrained readouts to reconstruct unseen chaotic dynamics. Using classical Lagrange interpolation theory, we derive rigorous bounds on readout errors and conjugacy residuals. Numerical examples demonstrate accurate prediction of chaotic dynamics across bifurcations.
\end{abstract}

\section{Introduction}
\label{sec:intro}
Scientific machine learning increasingly complements equation-based models in diverse fields of physics.
In this context, a central challenge in nonlinear science is to predict dynamics across bifurcations -- parameter-dependent transitions between steady, oscillatory, and chaotic behavior.
For example, extrapolating learned models of turbulent dynamics to unseen flow conditions is critical because training data are typically available only for a very limited range of conditions~\cite{Inubushi2020}.
The aim of the present paper is to establish a theoretical foundation for {this extrapolation capability.}

As a concrete model of machine learning suitable for nonlinear dynamics, we focus on reservoir computing (RC).
In this framework, only the readout weights of a driven recurrent network are trained while the internal weights {remain fixed.}
{Feeding the network output back as input produces an autonomous predictor}~\cite{Pathak2017,Pathak2018,Lu2018}, reconstructing ergodic properties, including Lyapunov spectra~\cite{Kobayashi2021}.
In pioneering work on extrapolation, Kim \emph{et al.}~\cite{Kim2021} demonstrated numerically that RC with additional control inputs can predict dynamics across bifurcations. 
However, the mechanism by which training on pre-bifurcation dynamics enables accurate reconstruction of unobserved bifurcation structure is not yet understood.

Generalized synchronization (GS) relates the target state to the driven reservoir state.
Based on GS, mathematical foundations for synchronization, embeddings, and approximation were established~\cite{Hart2020,Grigoryeva2021}.
In Hara and Kokubu's invariant-section formulation, 
the inverse of an injective GS map defines an ideal readout, yielding autonomous reservoir dynamics conjugate to the target dynamics~\cite{Hara2024,Hara2026}.
This rapidly developing GS-based framework has also led to new applications, including the theoretical justification of nonlinear readouts~\cite{Ohkubo2024}.

Here, we apply the GS-based framework to the prediction of dynamics across bifurcations. Specifically, studying the parameter dependence of ideal readouts, we derive a simple method that linearly combines readouts pretrained independently for the same fixed reservoir.
Figure~\ref{fig:concept} illustrates this method using the Lorenz system as an example.
Combining this framework with classical Lagrange interpolation theory, 
we derive rigorous bounds on readout errors and the associated conjugacy residuals. Using examples ranging from saddle-node and pitchfork normal forms to the Lorenz system, we demonstrate that the proposed method accurately predicts dynamics across bifurcations.
Finally, we discuss a possible link between this extrapolation capability and the identity theorem for analytic functions.

\begin{figure}[bhtp]
 \centering
 \includegraphics[width=\columnwidth]{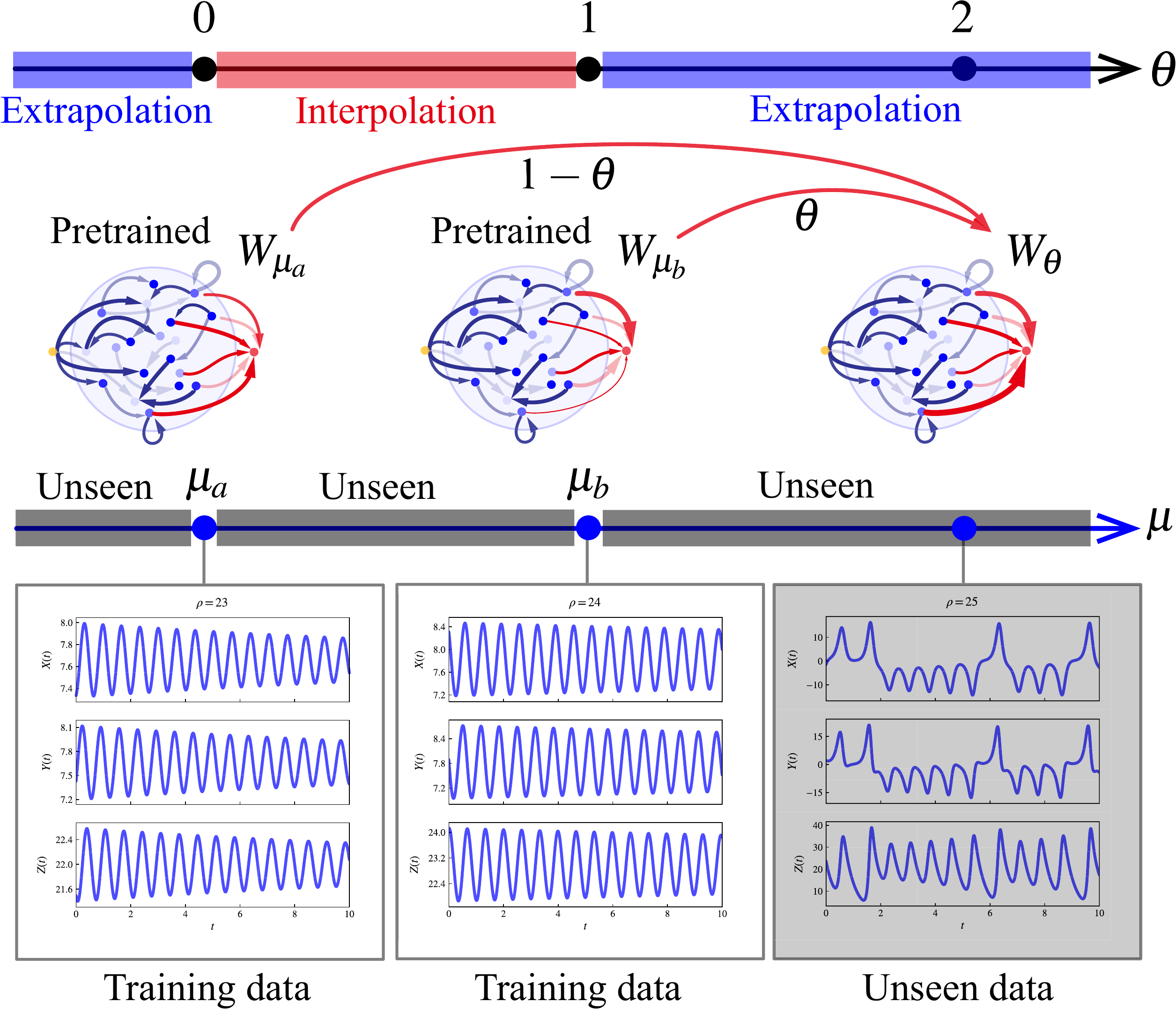}
 \caption{{Concept of readout interpolation and extrapolation with a linear combination of pretrained readout weights
 $W_\theta=(1-\theta)W_{\mu_a}+\theta W_{\mu_b}$.
 The same fixed reservoir is used throughout; only the readout weights change. Weights $W_{\mu_a}$ and $W_{\mu_b}$ are trained at $\mu_a$ and $\mu_b$. The mixing coefficient $\theta$ gives interpolation for $0\leq\theta\leq1$ and extrapolation otherwise. The lower axis uses the Lorenz parameter $\mu=\rho$, with $\mu_a=23$, $\mu_b=24$, and comparison value $\tilde\mu(2)=25$ under the parameter-correspondence ansatz in Eq.~\eqref{eq:nominal}.}}
 \label{fig:concept}
\end{figure}

\section{Synchronization and ideal reconstruction}
\label{sec:method}
{Let $J\subset\R$ be an open interval of the bifurcation parameter.} Let $\mathcal X:=\R^{d_x}$ denote the state space of the dynamical systems to be learned, referred to as target dynamical systems, and {let $\mathcal R \subset \R^{d_r}$ denote the state space of the recurrent neural network, referred to as reservoir space.} Write $\norm{\cdot}_{\mathcal X}$ and $\norm{\cdot}_{\mathcal R}$ for their Euclidean norms, and $\norm{\cdot}_{\mathrm{op}}$ for the induced operator norm. For each $\mu\in J$, let $M_\mu\subset\mathcal X$ be a nonempty compact set and $f_\mu:M_\mu\to M_\mu$ a continuous map. A continuous reservoir map $F:\mathcal X\times\mathcal R\to\mathcal R$, independent of $\mu$, defines
\begin{equation}
\begin{aligned}
 x_{n+1}&=f_\mu(x_n),\\
 r_{n+1}&=F(x_n,r_n).
\end{aligned}
\label{eq:driven}
\end{equation}
For example, $F(x,r)=\tanh(Ar+Bx+d)$, with componentwise activation and fixed $A,B,d$.
$A$ is the recurrent weight matrix, $B$
is the input weight matrix,
and $d$ is the input bias vector.
In a standard echo state network (ESN), $A$, $B$, and $d$ are randomly chosen and kept fixed, while only the readout weights are trained.

Assume that, for every $\mu\in J$, the same reservoir $F$ admits a continuous injective invariant section (GS map) $s_\mu:M_\mu\to\mathcal R$, satisfying
\begin{equation}
 s_\mu(f_\mu(x))=F(x,s_\mu(x)),\qquad x\in M_\mu.
\label{eq:gs}
\end{equation}
Define the graph
\begin{equation}
 \mathcal G[s_\mu]:=\{(x,r)\in M_\mu\times\mathcal R:r=s_\mu(x)\}.
\end{equation}
Equation~\eqref{eq:gs} expresses its forward invariance:
\begin{equation}
 (x_n,r_n)\in\mathcal G[s_\mu]
 \quad\Longrightarrow\quad
 (x_{n+1},r_{n+1})\in\mathcal G[s_\mu].
\end{equation}
For the ideal driven dynamics, we consider initial states on $\mathcal G[s_\mu]$. 
Theorem~2.4 of Ref.~\onlinecite{Hara2026} gives sufficient conditions for differentiable invariant sections.
{Since $M_\mu$ is compact, the continuous injection $s_\mu$ is a homeomorphism onto its compact image $s_\mu(M_\mu)$.}
We call the inverse GS map the \emph{ideal readout},%
\begin{equation}
 \omega_\mu=s_\mu^{-1}:s_\mu(M_\mu)\longrightarrow M_\mu.
\label{eq:inverse}
\end{equation}
The ideal autonomous reservoir is obtained by closing the loop with the ideal readout $\omega_\mu=s_\mu^{-1}$:
\begin{equation}
 G_\mu(r)=F(\omega_\mu(r),r),\qquad r\in s_\mu(M_\mu).
\label{eq:ideal}
\end{equation}
Since $\omega_\mu(s_\mu(x))=x$, Eq.~\eqref{eq:gs} gives
\begin{equation}
 G_\mu\circ s_\mu=s_\mu\circ f_\mu.
\label{eq:conjugacy}
\end{equation}
Thus $G_\mu|_{s_\mu(M_\mu)}$ is conjugate to $f_\mu$, by the same identity used in Theorem~2.24(2) of Ref.~\onlinecite{Hara2026}.

\section{Derivation of the linear combination formula}
\label{sec:continuation}
We first consider the family of ideal readouts (Step 1), then interpolate them (Step 2), and finally substitute trained readouts (Step 3).

\subsection{\underline{Step 1}: Ideal readouts on a common reservoir region}
\label{sec:stage1}
Let $V\subset\mathcal R$ be a compact set containing $\bigcup_{\mu\in J}s_\mu(M_\mu)$. Assume continuous extensions to $V$, denoted by the same symbols $\omega_\mu$:%
\begin{equation}
 \omega_\mu\in\mathcal B:=C^0(V;\mathcal X),\qquad
 \omega_\mu|_{s_\mu(M_\mu)}=s_\mu^{-1}.
\label{eq:extension}
\end{equation}
Here $C^0(V;\mathcal X)$ denotes the continuous $\mathcal X$-valued functions on $V$. Use the uniform norm
\begin{equation}
 \norm{u}_{\mathcal B}=\sup_{r\in V}\norm{u(r)}_{\mathcal X}.
\label{eq:norm}
\end{equation}
Since $V$ is compact and $u$ is continuous, the supremum is attained and can also be written as a maximum.

\subsection{\underline{Step 2}: Interpolation of the ideal readouts}
\label{sec:stage2}
For two parameters $\mu_a<\mu_b$ in $J$, fix a bounded closed interval $\Theta\supseteq[0,1]$ and choose the degree-one interpolant
\begin{equation}
 \omega^{\mathrm{lin}}_\theta
   =(1-\theta)\omega_{\mu_a}+\theta\omega_{\mu_b},
 \qquad \theta\in\Theta.
\label{eq:ideal_interpolation}
\end{equation}
Equation~\eqref{eq:ideal_interpolation} gives interpolation for $0\leq\theta\leq1$ and extrapolation otherwise.

To specify a comparison target, we adopt the following affine ansatz for the parameter correspondence:
\begin{equation}
 h:=\mu_b-\mu_a,\qquad
 \tilde{\mu}(\theta)=\mu_a+\theta h.
\label{eq:nominal}
\end{equation}
Choose $\Theta$ so that $I_\Theta:=\tilde{\mu}(\Theta)\subseteq J$.
This ansatz specifies the comparison between target and reservoir dynamics. We test it numerically for the Lorenz system in {Fig.~\ref{fig:lorenz_results}(d)}.

{Proposition~\ref{prop:idealerror} gives an error estimate for the interpolant using classical degree-one Lagrange interpolation theory. See Supplemental Appendix E for the proof.}

\begin{proposition}[Interpolation and extrapolation error]
\label{prop:idealerror}
Assume that the map $\mu\mapsto\omega_\mu\in\mathcal B$ is $C^2$ on $J$. Define
\begin{equation}
 C:=\sup_{\nu\in I_\Theta}
        \norm{\frac{d^2\omega_\nu}{d\nu^2}}_{\mathcal B}<\infty.
\label{eq:curvature}
\end{equation}
Then, for every $\theta\in\Theta$,
\begin{equation}
 \norm{\omega^{\mathrm{lin}}_\theta-\omega_{\tilde{\mu}(\theta)}}_{\mathcal B}
 \leq \frac{C}{2}h^2|\theta(\theta-1)|.
\label{eq:idealbound}
\end{equation}
\end{proposition}

\subsection{\underline{Step 3}: Readout approximation and the weight formula}
\label{sec:stage3}
{The ideal readouts are generally unavailable in practice. In the conventional RC framework, we approximate them at the two training parameters by trained linear readouts for the same reservoir dynamics $F$,}
\begin{equation}
 \omega_{\mu_i}(r)\approx\widehat\omega_{\mu_i}(r)
        :=W_{\mu_i} r,\qquad i=a,b.
\label{eq:readouts}
\end{equation}
{As in conventional RC~\cite{Pathak2017,Pathak2018,Kim2021,Hara2026}, each $W_{\mu_i}$ is obtained by linear least squares,}

\begin{equation}
 W_{\mu_i}=\arg\min_W\sum_{n\in\mathcal T_i}\norm{Wr_{n+1}-x_{n+1}}_{\mathcal X}^2.
 \label{eq:least_squares}
\end{equation}

{Here $\mathcal T_i$ is the set of sample indices used for training at $\mu_i$.} The subscript $\mu_i$ identifies a training parameter, whereas $\theta$ in $W_\theta$ denotes the mixing coefficient.

Substituting these trained readouts into Eq.~\eqref{eq:ideal_interpolation} gives the \emph{combined readout}
\begin{equation}
\begin{aligned}
 \widehat\omega_\theta(r)
 &:=(1-\theta)\widehat\omega_{\mu_a}(r)+\theta\widehat\omega_{\mu_b}(r)\\
  &=\bigl[(1-\theta)W_{\mu_a}+\theta W_{\mu_b}\bigr]r.
\end{aligned}
\label{eq:interpolation}
\end{equation}
{The weight formula and the autonomous update obtained by feeding the combined readout back as input are therefore}
\begin{equation}
\begin{aligned}
 W_\theta&=(1-\theta)W_{\mu_a}+\theta W_{\mu_b},\\
 r_{n+1}&=\widehat G_\theta(r_n):=F(W_\theta r_n,r_n).
\end{aligned}
\label{eq:algorithm}
\end{equation}
The linear-interpolation structure of Eq.~\eqref{eq:algorithm} resembles that arising in the control-input analysis of Tadokoro \emph{et al.}, although the derivation and assumptions differ~\cite{Tadokoro2024}.

Let $e_i=\widehat\omega_{\mu_i}-\omega_{\mu_i}$. {The total readout error separates into the interpolation error and the errors in the two pretrained readouts:}
\begin{equation}
\begin{aligned}
 \widehat\omega_\theta-\omega_{\tilde{\mu}(\theta)}
 &=\widehat\omega_\theta+\omega^{\mathrm{lin}}_\theta
       -\omega^{\mathrm{lin}}_\theta-\omega_{\tilde{\mu}(\theta)}\\
 &=\bigl(\omega^{\mathrm{lin}}_\theta-\omega_{\tilde{\mu}(\theta)}\bigr)
       +\bigl(\widehat\omega_\theta-\omega^{\mathrm{lin}}_\theta\bigr)\\
 &=\bigl(\omega^{\mathrm{lin}}_\theta-\omega_{\tilde{\mu}(\theta)}\bigr)
       +(1-\theta)e_a+\theta e_b.
\end{aligned}
\label{eq:error_decomposition}
\end{equation}
\begin{corollary}[Error after readout approximation]
\label{cor:implemented}
Under the assumptions of Proposition~\ref{prop:idealerror}, set $\varepsilon_i=\norm{e_i}_{\mathcal B}$. For every $\theta\in\Theta$,
\begin{equation}
\begin{aligned}
 \norm{\widehat\omega_\theta-\omega_{\tilde{\mu}(\theta)}}_{\mathcal B}
 \leq{}&|1-\theta|\varepsilon_a+|\theta|\varepsilon_b\\
 &+\frac{C}{2}h^2|\theta(\theta-1)|\\
 =:{}&E(\theta).
\end{aligned}
\label{eq:bound}
\end{equation}
\end{corollary}
\par\indent\emph{Proof---}\ignorespaces
Take norms in Eq.~\eqref{eq:error_decomposition} and apply Proposition~\ref{prop:idealerror}.\qedmark

For $i=a,b$, $\varepsilon_i$ measures the uniform approximation error over the entire common region $V$, not only on the synchronized set $s_{\mu_i}(M_{\mu_i})$ at the corresponding training parameter. A small approximation error on this set alone does not guarantee a small $\varepsilon_i$.

Conjugacy is a central concept for assessing dynamical reconstruction. Related conjugacy errors have been evaluated numerically~\cite{Hara2024,Ohkubo2024}. Here, we derive an upper bound on the conjugacy residual defined below.
\label{sec:conjugacy_residual}
\leavevmode\par\nopagebreak
\begin{corollary}[Conjugacy-residual bound]
\label{cor:residual}
{Under the assumptions of Corollary~\ref{cor:implemented}, put $\mu=\tilde\mu(\theta)$. Suppose $F$ is Lipschitz in its input between $x$ and $W_\theta s_\mu(x)$ at reservoir state $s_\mu(x)$, with one constant $L_x$ for all $x\in M_\mu$ and $\theta\in\Theta$. Define}
\begin{equation}
 \mathcal R_\theta(x):=\widehat G_\theta(s_\mu(x))-s_\mu(f_\mu(x)),
 \quad x\in M_\mu.
 \label{eq:conjugacy_residual}
\end{equation}
Then
\begin{equation}
 \sup_{x\in M_\mu}\norm{\mathcal R_\theta(x)}_{\mathcal R}\leq L_xE(\theta).
 \label{eq:conjugacy_bound}
\end{equation}
\end{corollary}
{The proof is given in Supplemental Appendix E.}

{The residual quantifies the defect in the conjugacy relation~\eqref{eq:conjugacy}.} For $F(x,r)=\tanh(Ar+Bx+d)$, it holds with $L_x=\norm{B}_{\mathrm{op}}$ because $|\tanh'|\leq1$. Supplementary Appendix C gives a fixed-point-position bound under additional local conditions.

\section{Numerical illustrations}
\label{sec:numerics}
{Here, we demonstrate that the proposed method enables extrapolation beyond the parameter values used for training to reconstruct fixed points and chaotic dynamics of the Lorenz system (Fig.~\ref{fig:lorenz_results}).}
{Results for scalar normal forms, including reconstructed fixed-point branches and stability changes, are presented in Appendix (Fig.~\ref{fig:fixed_points}).}
For comparison with earlier work, we consider the same target systems as Kim \emph{et al.}~\cite{Kim2021}.
The methodological differences between the two approaches are summarized in Appendix.

\subsection{Lorenz equations}
{We next examine the Lorenz system, {$\dot X=10(Y-X)$, $\dot Y=\rho X-Y-XZ$, and $\dot Z=XY-\tfrac83 Z$,} using the ESN and numerical settings in Supplemental Appendix A. 
Following Kim \emph{et al.}~\cite{Kim2021}, training uses four transient trajectories: at each of $\rho_a=23$ and $\rho_b=24$, one trajectory converges to each of the two nonzero stable fixed points {(Fig.~\ref{fig:concept})}. 
{We use only data at $\rho=23$ and $24$ for training the readout weights and for Bayesian optimization of the hyperparameters.}
The learned weights $W_{23}$ and $W_{24}$ give}
$
 W_\theta=(1-\theta)W_{23}+\theta W_{24}.
$
{{The parameter-correspondence ansatz gives the reference parameter $\tilde\rho(\theta)=23+\theta$.} We first use this pairing for the fixed-point comparison, and then compare predicted and reference distributions while varying $\theta$ and $\rho$ independently {({Fig.~\ref{fig:lorenz_results}(d)})}.}

\begin{figure}[!t]
 \centering
 \includegraphics[width=\textwidth]{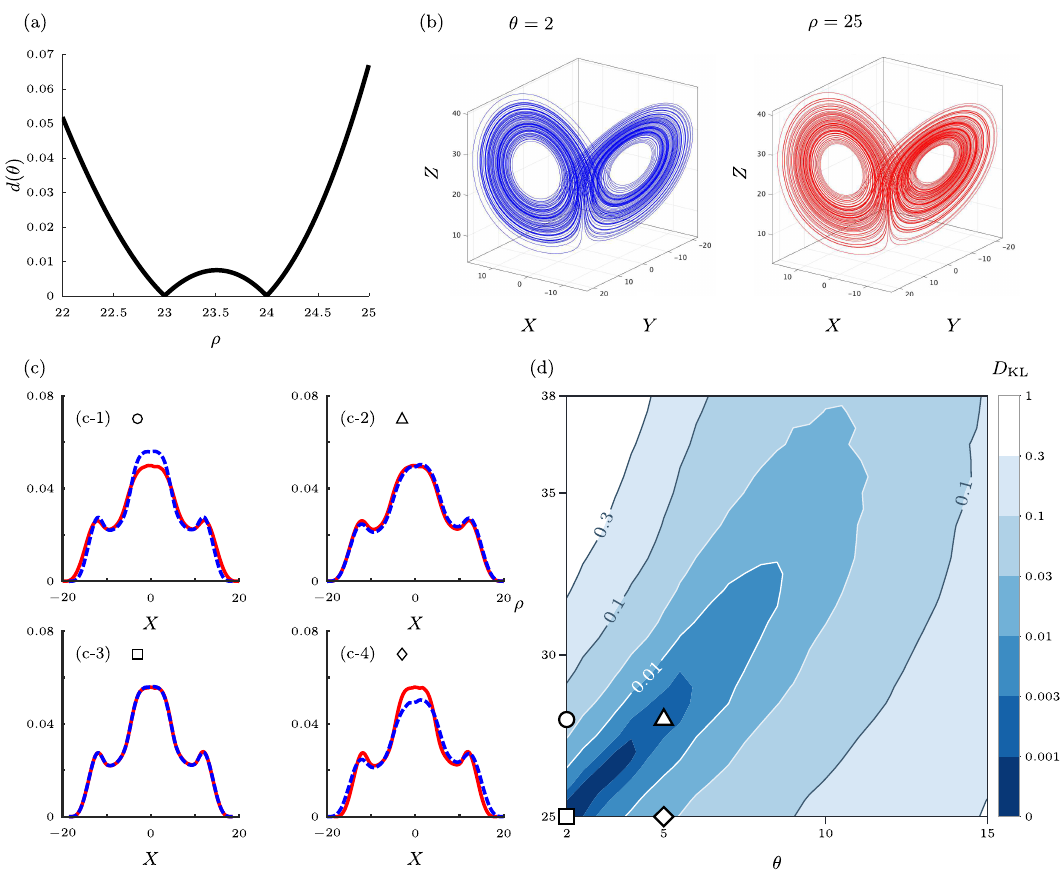}
 \caption{{Reconstruction of Lorenz dynamics. (a) Error in the reconstructed fixed-point position, $d(\theta)$ [Eq.~\eqref{eq:fixedpoint_distance}], plotted against $\tilde\rho(\theta)=23+\theta$ on the $\rho$ axis. The black curve does not encode stability. Interpolation corresponds to $23\leq\rho\leq24$; the remaining range is extrapolation. (b) Predicted trajectory at $\theta=2$ (left, blue) and reference trajectory obtained by numerically integrating the Lorenz equations at $\rho=25$ (right, red). (c) Estimated marginal probability densities of $X$ from predicted (blue dashed) and reference (red solid) trajectories: (c-1) $(\theta,\rho)=(2,28)$, (c-2) $(5,28)$, (c-3) $(2,25)$, and (c-4) $(5,25)$, in reading order. (d) Kullback--Leibler divergence $D_{\mathrm{KL}}(P_\theta\|Q_\rho)$ between the predicted and reference marginal distributions of $X$, versus $\theta$ and $\rho$. The circle, triangle, square, and diamond identify (c-1)--(c-4), respectively. Dark blue denotes smaller divergence and white larger divergence.}}
  \label{fig:lorenz_results}
\end{figure}

\subsection{Fixed-point positions}
{{Figure~\ref{fig:lorenz_results}(a)} plots the error $d(\theta)$ in the reconstructed fixed-point position against $\rho$, with the reconstructed results labeled by $\tilde\rho(\theta)=23+\theta$. The error is the Euclidean distance between the reconstructed position $\widehat x_\theta^*$ and the positive-$X$ Lorenz equilibrium $x_+^*(\rho)=\bigl(\sqrt{\tfrac83(\rho-1)},\,\sqrt{\tfrac83(\rho-1)},\,\rho-1\bigr)$:}

\begin{equation}
 d(\theta)=\bigl\|\widehat x_\theta^*-{x_+^*(\tilde\rho(\theta))}\bigr\|_{\mathcal X}.
\label{eq:fixedpoint_distance}
\end{equation}
{Here $\widehat x_\theta^*:=W_\theta\widehat r_\theta^*$. {As in Fig.~\ref{fig:fixed_points}(a,b), all reconstructed fixed points} are computed using Newton's method in reservoir space, and mapped to the target state space through the combined readout. The range $-1\leq\theta\leq2$ corresponds to $22\leq\tilde\rho(\theta)\leq25$.}

{Corollary~\ref{cor:residual}, together with the additional local fixed-point conditions of Supplemental Appendix C, uniformly valid on $[0,1]$, gives}
\begin{equation}
\begin{aligned}
 d(\theta)&\leq C_*\left[(1-\theta)\varepsilon_a+\theta\varepsilon_b
             +\frac{C}{2}h^2\theta(1-\theta)\right]{,}\\
 &\hspace{12em}{0\leq\theta\leq1.}
\end{aligned}
\label{eq:uniform_position_main}
\end{equation}
{Here $C_*:=1+\kappa L_x\max\{\norm{W_{\mu_a}}_{\mathrm{op}},\norm{W_{\mu_b}}_{\mathrm{op}}\}${, where $\kappa$ uniformly bounds the distance to $\widehat r_\theta^*$ by the fixed-point residual: $\norm{r-\widehat r_\theta^*}_{\mathcal R}\leq\kappa\norm{r-\widehat G_\theta(r)}_{\mathcal R}$ in a neighborhood containing the corresponding ideal fixed point $s_{\tilde\rho(\theta)}(x_+^*(\tilde\rho(\theta)))$.} Linear asymptotic stability gives a sufficient local condition; the proof is in Supplemental Appendix C.}

{The error in {Fig.~\ref{fig:lorenz_results}(a)} is small at the two training parameter values $\rho=23,24$ and has an approximately parabolic profile between them, consistent with the bound above. For $\tilde\rho(\theta)<23$ or $\tilde\rho(\theta)>24$, the error increases away from the training parameter values.}

\subsection{{Extrapolation of chaotic dynamics}}
{{Figure~\ref{fig:lorenz_results}(b)} compares a predicted trajectory at $\theta=2$ with a reference trajectory at $\rho=25$ in the $(X,Y,Z)$ phase space. The prediction is obtained by iterating $\widehat G_2$ and applying $W_2$, and the reference by numerically integrating the Lorenz equations. The parameter-correspondence ansatz gives $\tilde\rho(2)=25$, beyond both training parameter values. Both trajectories exhibit similar butterfly-shaped geometry.}

{{Figure~\ref{fig:lorenz_results}(c)} compares estimates of the marginal probability density of $X$ for four pairs $(\theta,\rho)$. The densities are estimated from sampled values of $X$ along the predicted and reference trajectories. For the prediction, $X$ is the first component of $W_\theta r_n$. Panels (c-2) and (c-3), $(\theta,\rho)=(5,28)$ and $(2,25)$, satisfy $\rho=\tilde\rho(\theta)$ and show close agreement. Panels (c-1) and (c-4), $(2,28)$ and $(5,25)$, exchange the reference parameters and reveal visible discrepancies, particularly near the central peak. The four panels therefore distinguish comparisons based on the ansatz from comparisons at the exchanged reference parameters.}

{{Figure~\ref{fig:lorenz_results}(d)} shows the Kullback--Leibler (KL) divergence $D_{\mathrm{KL}}(P_\theta\|Q_\rho)$ over independently varied mixing coefficients $\theta$ (horizontal axis) and reference parameters $\rho$ (vertical axis). Here $P_\theta$ is the marginal distribution of $X$ estimated from predicted trajectories at mixing coefficient $\theta$, and $Q_\rho$ is the corresponding distribution estimated from reference trajectories at parameter $\rho$. For discretized distributions, $D_{\mathrm{KL}}(P_\theta\|Q_\rho)=\sum_j P_\theta(j)\log[P_\theta(j)/Q_\rho(j)]$, where $j$ labels common bins, the probabilities in each distribution sum to one, and $\log$ is the natural logarithm. The circle, triangle, square, and diamond correspond to panels (c-1), (c-2), (c-3), and (c-4), respectively; the dark-blue band identifies pairs with closer marginal distributions. For integer $\rho=25,26,27,28$, minima over the sampled $\theta$ values occur at $\theta=\rho-23$, with divergences $6.05\times10^{-5}$, $4.56\times10^{-4}$, $5.18\times10^{-4}$, and $1.54\times10^{-3}$, respectively. Deviations from the parameter-correspondence ansatz occur farther into extrapolation; the ansatz does not guarantee exact dynamical agreement. Supplemental Fig.~S1 shows all 24 comparisons of estimated marginal probability densities for $\theta=2,3,\ldots,7$ and $\rho=25,26,27,28$. The evaluated parameter grid and a comparison with an implementation using an additional control input are described in Supplemental Appendix B.}

{As a reference value, we use the KL divergence between marginal distributions of $X$ estimated from two reference trajectories obtained by numerically integrating the Lorenz equations at the same $\rho$ from different initial conditions. This value serves as a practical lower bound for interpreting agreement between distributions estimated from finite samples. {The measured reference value is $7.09\times10^{-5}$.}}

\section{Conclusion and discussion}
\label{sec:discussion}
{We have developed a GS-based framework for extrapolating learned dynamics across bifurcations by combining pretrained RC readouts, with rigorous bounds on readout errors and conjugacy residuals. The method reconstructs both the stable and untrained unstable saddle-node branches [Fig.~\ref{fig:fixed_points}(a)] and the pitchfork's loss of stability and emergence of two stable branches beyond the training parameter values [Fig.~\ref{fig:fixed_points}(b)]. Lorenz fixed-point errors are consistent with the conditional position bound [{Fig.~\ref{fig:lorenz_results}(a)}], while {predicted and reference trajectories have similar geometry, and their estimated marginal probability densities of $X$ are close for the selected parameter pairs} [{Fig.~\ref{fig:lorenz_results}(b,c)}]. The KL comparison tests the parameter-correspondence ansatz and reveals deviations farther into extrapolation [{Fig.~\ref{fig:lorenz_results}(d)}].}

\subsection{Extrapolation capability and the identity theorem}
Here, we discuss a possible link between this extrapolation capability and the identity theorem for analytic functions
under an idealized setting that differs substantially from our numerical experiments.
Suppose that the readout error \(\widehat\omega_\theta-\omega_{\tilde\mu(\theta)}\) depends real-analytically on \(\theta\) over a connected open interval and vanishes identically on \(V\) for every \(\theta\) in a nonempty open subinterval.
The identity theorem then implies that the error vanishes throughout the entire parameter interval, yielding conjugacy between the target and reconstructed dynamics at each corresponding parameter value (Supplemental Appendix F). The relevance of this idealized argument to extrapolation from finite training data remains conjectural.

{Even when a learned model reconstructs the chaotic dynamics and long-time statistics accurately, {{a predicted trajectory will eventually separate} from the corresponding trajectory of the target system} because of the intrinsic sensitivity to initial conditions. {Maintaining accurate predictions of individual trajectories over extended periods} may therefore benefit from combining the learned model with data assimilation designed around the instability structure of the nonlinear dynamics~\cite{Inubushi2023}.}

\section*{Supplementary material}
{See Supplemental Material~\cite{SupplementalMaterial} for numerical settings, KL-divergence details, the fixed-point bound, scalar examples, proofs, the conditional analytic-continuation {argument}{, and the full set of comparisons of marginal probability densities.}}

\clearpage
\section*{Appendix}
\setcounter{figure}{2}
\begin{figure}[!htbp]
 \centering
 \includegraphics[width=\textwidth]{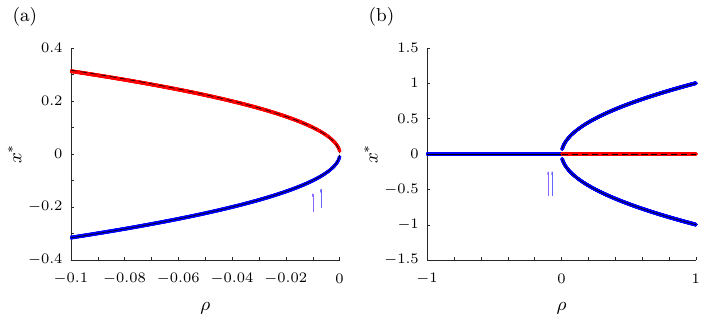}
 \caption{{Reconstruction of fixed points of scalar normal forms. (a) Saddle-node, with readouts trained at $\rho_a=-0.01$ and $\rho_b=-0.007$. (b) Supercritical pitchfork, with readouts trained at $\rho_a=-0.1$ and $\rho_b=-0.07$. Solid/dashed black curves denote stable/unstable fixed-point branches of the target equations; blue/red curves show the readout values of stable/unstable fixed points of the autonomous map $\widehat G_\theta$. Arrows mark the two parameter values at which training data are available. The target parameter $\rho$ and the reconstructed results labeled by $\tilde\rho(\theta)=\rho_a+\theta(\rho_b-\rho_a)$ share the same horizontal axis.}}
 \label{fig:fixed_points}
\end{figure}

\subsection{{Scalar normal forms}}
\begingroup
{For the saddle-node normal form $\dot x=\rho+x^2$, Fig.~\ref{fig:fixed_points}(a) compares the positions and linear stability of the reference and reconstructed fixed points as functions of $\rho$. Readouts are trained only at the two parameter values $\rho_a=-0.01$ and $\rho_b=-0.007$ using transients converging to the stable branch $x_-^*=-\sqrt{-\rho}$. We compute fixed points $\widehat r_\theta^*$ of the autonomous map $\widehat G_\theta$ using Newton's method in reservoir space and plot their readout values $W_\theta\widehat r_\theta^*$. Their linear stability is determined from the eigenvalues of $D_r\widehat G_\theta(\widehat r_\theta^*)$, the Jacobian of the autonomous map evaluated at each fixed point. The reconstructed branches are plotted at $\tilde\rho(\theta)=\rho_a+\theta(\rho_b-\rho_a)$ on the $\rho$ axis, following the parameter-correspondence ansatz. Although the training trajectories approach only the attracting branch, the autonomous map with the combined readout also reconstructs the unstable branch $x_+^*=\sqrt{-\rho}$ and tracks both branches toward the saddle-node bifurcation at $\rho=0$, beyond both training parameter values.}

{For the supercritical pitchfork normal form $\dot x=\rho x-x^3$, Fig.~\ref{fig:fixed_points}(b) compares the corresponding fixed-point positions and stability. Readouts are trained only at the two parameter values $\rho_a=-0.1$ and $\rho_b=-0.07$ using transients converging to the stable origin. {We use the same calculation and parameter-correspondence ansatz as in the saddle-node example.} Extrapolation reproduces the loss of stability of $x^*=0$ and the emergence of the two stable branches $x^*=\pm\sqrt{\rho}$ for $\rho>0$. These examples therefore demonstrate reconstruction of fixed-point branches and stability changes absent from the training data, rather than fixed-point locations alone. Detailed numerical settings and continuation procedures are given in Supplemental Appendix D.}
\endgroup

\subsection{{Relation to control-input learning by Kim {\it et al.} (2021)}}
\label{sec:kim_appendix}
{Kim \emph{et al.}~\cite{Kim2021} fit a single readout to concatenated reservoir states and target trajectories generated with an additional control input. Our method differs in three respects. First, independently pretrained readouts for the same fixed reservoir can be reused without retraining. 
Second, we vary the predicted dynamics by adjusting only the readout weights, without introducing an additional control input to the reservoir.
This is suitable for implementations of RC using physical systems~\cite{NakajimaFischer2021}, because incorporating such an input into their dynamics can be difficult. Our method allows the physical reservoir itself to remain unchanged, provided that the readout weights are adjustable and the output can be fed back as input.
Third, at mixing coefficients $\theta=0$ and $\theta=1$, the identities $W_0=W_{\mu_a}$ and $W_1=W_{\mu_b}$ recover the readouts optimized independently at $\mu_a$ and $\mu_b$, respectively. Each therefore retains the minimum least-squares error for the reservoir states and target outputs at its training parameter value [Eq.~\eqref{eq:least_squares}]. In contrast, a readout fitted jointly to all parameter values minimizes the sum of their regression losses
and need not minimize the loss for each parameter value separately~\cite{Kim2021}.}

\end{document}